\documentclass[a4paper,twocolumn]{article}

\usepackage[margin=20mm]{geometry}

\usepackage{hyperref}
\usepackage{amsmath,amssymb,amsfonts}

\newtheorem{theorem}{Theorem}
\newenvironment{proof}{\paragraph{Proof:}}{\hfill$\square$}
\newtheorem{property}{Property}
\newtheorem{proposition}{Proposition}

\newtheorem{definition}{Definition}
\newtheorem{problem}{Problem}

\usepackage{stmaryrd} 
\usepackage{graphicx}
\usepackage{tikz} \usetikzlibrary{arrows.meta}

\usepackage{subcaption}
\usepackage{algorithm}
\usepackage[noend]{algpseudocode}
\algloopdefx[loop]{Iterate}
[1][]{{\bf iterate} }
\newcommand{\Output}{\State {\bf output}~}

\newcommand{\ltok}[1]{{#1}}

\newcommand{\flok}[1]{{#1}}

\newcommand{\DPM}{\texorpdfstring{\ensuremath{C^{+-}}}{C+-}}
\newcommand{\DPP}{\texorpdfstring{\ensuremath{C^{++}}}{C++}}
\newcommand{\APM}{{\tt A+-}}
\newcommand{\APP}{{\tt A++}}
\newcommand{\AMM}{{\tt A--}}
\newcommand{\adjonly}{mere}
\newcommand{\adjfull}{full}
\newcommand{\listonly}{\adjonly{}-listing}
\newcommand{\listfull}{\adjfull{}-listing}
\newcommand{\SAT}{NAE3SAT+}

\newcommand{\codegit}{https://github.com/lecfab/volt}

\begin{document}

\title{Tailored vertex ordering for faster triangle~listing~in~large~graphs}

%
\newcommand{\sorbonne}{Sorbonne Université, CNRS, LIP6, F-75005 Paris, France}
\author{Fabrice Lécuyer\thanks{\sorbonne}
\and Louis Jachiet\thanks{LTCI, Télécom Paris, Institut Polytechnique de Paris}
\and Clémence Magnien$^*$
\and Lionel Tabourier$^*$}

\date{}


\maketitle

\begin{abstract}
Listing triangles is a fundamental graph problem with many applications, and large graphs require fast algorithms.
%
Vertex ordering allows the orientation of edges from lower to higher vertex indices, 
%
and state-of-the-art triangle listing algorithms use this to accelerate their execution and to bound their time complexity. Yet, only basic orderings have been tested.
%
In this paper, we show that studying the precise cost of algorithms instead of their bounded complexity  leads to faster solutions.
We introduce cost functions that link ordering properties with the running time of a given algorithm.
We prove that their minimization is NP-hard and propose heuristics to obtain new orderings with different trade-offs between cost reduction and ordering time.
Using datasets with up to two billion edges, we 
show that
our heuristics accelerate the listing of triangles by an average of 38\% when the ordering is already given as an input,
and 16\% when the ordering time is included.\\
(arxiv version: \href{https://doi.org/10.48550/arXiv.2203.04774}{doi.org/10.48550/arXiv.2203.04774})
\end{abstract}



\section{Introduction}

\subsection{Context and problem.}

Small connected subgraphs are key to identifying families of real-world networks~\cite{milo2002network}
and are used for descriptive or predictive purposes in various fields such as biology~\cite{sporns2004motifs,prvzulj2007biological}, linguistics~\cite{biemann2016network} or engineering~\cite{valverde2005network}.
%
%
In sociology in particular, characterizing networks with specific structural patterns has been a focus of interest for a long time, as it is even present in the works of early 20$^{th}$ century sociologists such as Simmel~\cite{simmel1908soziologie}.
Consequently, it is a common practice in social network analysis to describe interactions between individuals using local patterns~\cite{holland1976local,wasserman1994social}. 
Recently, the ability to count and list small size patterns efficiently allowed the characterization of various types of social networks on a large scale~\cite{choobdar2012comparison,charbey2019stars}. 
In particular, listing elementary motifs such as triangles and 3-motifs is a stepping stone in the analysis of the structure of networks and their dynamics~\cite{faust2010puzzle}.
For instance, the closure of a triplet of nodes to form a triangle is supposed to be a driving force of social networks evolution~\cite{leskovec2008microscopic,sintos2014using}.

The task of listing triangles may seem simple, but web crawlers and social platforms generate graphs that are so large that scalability becomes a challenge.
%
%
Thus, a lot of effort has been dedicated to efficient in-memory triangle listing.
Note that methods exist for graphs that do not fit in main memory: some use I/O-efficient accesses to the disk~\cite{chu2011triangle}, while others partition the graph and process each part separately~\cite{arifuzzaman2019fast}.
However, such approaches induce a costly counterpart that makes them much less efficient than in-memory listing methods.
%
%
%
It is also worth noticing that exact or approximate methods designed for triangle \textit{counting}~\cite{al2018triangle,gpu_triangle,gpu_order} can generally not be adapted to triangle \textit{listing}.

An efficient algorithm for triangle listing has been proposed early on in~\cite{chiba1985arboricity}.
Based on the observation that real-world graphs generally have a heterogeneous degree distribution, 
later contributions~\cite{schank2005finding,latapy2008main} showed how ordering vertices by degree or core value accelerates the listing.
Such orderings create an orientation of edges so that nodes that are costly to process are not processed many times.
A unifying description of this method has been proposed in~\cite{ortmann2014} and it has been successfully extended to larger cliques~\cite{danisch2018listing,li2020cliques,uno2012cliques}.
However, only degree and core orderings have been exploited, but their properties are not specifically tailored for the triangle listing problem. 
Other types of orderings benefited other problems such as graph compression~\cite{boldi2004webgraph,dhulipala2016compressing}
or cache optimization~\cite{gorder,lee2019pre}.
The main purpose of this work is thus to find a general method to design efficient vertex orderings for triangle listing.

\subsection{Contributions.}
In this work, we show how vertex ordering directly impacts the running time of
the two fastest existing triangle listing algorithms. First, we introduce cost functions that relate the vertex ordering and the running time of each algorithm.
We prove that finding an optimal ordering that minimizes either of these costs is NP-hard.
Then, we expose a gap in the combinations of algorithm and ordering considered in the literature, and we bridge it with three heuristics producing orderings with low corresponding costs.
Our heuristics reach a compromise between their running time and the quality of the ordering obtained, in order to address two distinct tasks:
listing triangles with or without taking into account the ordering time.
Finally, we show that our resulting combinations of algorithm and ordering outperform state-of-the-art running times for either task.
We release an efficient open-source implementation\footnote{Open-source {\tt c++} implementation available at: \url{\codegit}} of all considered methods.

%
Section~\ref{review} presents state-of-the-art methods to list triangles. In Section~\ref{method},
we analyze the cost induced by a given ordering on these algorithms
and propose several heuristics to reduce it; the proofs of NP-hardness are in Appendix~\ref{appendix:dpm} and~\ref{appendix:dpp}.
The experiments of Section~\ref{experiments} show that our methods are efficient in practice and improve the state of the art.


\subsection{Notations.}
\label{sec:notations}
We consider an unweighted undirected simple graph $G=(V,E)$ with $n=|V|$ vertices and $m=|E|$ edges.
The set of neighbors of a vertex $u$ is denoted $N_u = \{v, \{u,v\} \in E\}$, and its degree is $d_u=|N_u|$. 
An ordering $\pi$ is a permutation over the vertices that gives a distinct index
$\pi_u\in\llbracket 1, n\rrbracket$ to each vertex $u$. In the directed
acyclic graph (DAG) $G_\pi=(V, E_\pi)$,
for $\{u,v\}\in E$, $E_\pi$
contains $(u,v)$ if $\pi_u<\pi_v$, and $(v,u)$
otherwise.
In such a directed graph, the set $N_u$ of neighbors of $u$ is partitioned into its predecessors $N^-_u$ and successors $N^+_u$.
We define the indegree $d^-_u=|N^-_u|$ and the outdegree $d^+_u=|N^+_u|$; their sum is $d^-_u+d^+_u=d_u$.
A triangle of $G$ is a set of vertices $\{u,v,w\}$ such that $\{u,v\},\{v,w\},\{u,w\}\in E$.
A $k$-clique is a set of $k$ fully-connected vertices.
%
%
The coreness $c_u$ of vertex $u$ is the highest value $k$ such that $u$ belongs to a subgraph of $G$ where all vertices have degree at least $k$;
the core value or degeneracy $c(G)$ of $G$ is the maximal $c_u$ for $u\in V$.
A core ordering $\pi$ verifies $\pi_u\leq\pi_v \Leftrightarrow c_u \leq c_v$.
\ltok{Core value and core ordering can be
computed in linear time~\cite{batagelj2003kcore}.}


\section{State of the art}\label{review}

\subsection{Triangle listing algorithms.}


Ortmann and Brandes~\cite{ortmann2014} have identified two families of triangle listing algorithms: \textit{adjacency testing}, and\textit{ neighborhood intersection}. 
The former sequentially considers each vertex $u$ as a seed, and processes all pairs $\{v,w\}$ of its neighbors;
if they are themselves adjacent, $\{u,v,w\}$ is a triangle. 
%
%
%
Algorithms
\texttt{tree-lister}~\cite{itai1978finding}, \texttt{node-iterator}~\cite{schank2005finding} and 
\texttt{forward}~\cite{schank2005finding}
belong to this category.
In contrast,
%
%
the neighborhood intersection family methods sequentially considers each edge $(u,v)$ as a seed; each common neighbor $w$ of $u$ and $v$ forms a triangle $\{u,v,w\}$. 
Algorithms 
\texttt{edge-iterator}~\cite{schank2005finding},
\texttt{compact-forward}~\cite{latapy2008main} and 
{\tt K3}~\cite{chiba1985arboricity} belong to this category, as well as some algorithms that list larger cliques~\cite{kanade2004cliques,danisch2018listing,li2020cliques}.
%


In naive versions of both adjacency testing and neighborhood intersection, finding a triangle $(u,v,w)$ does not prevent from finding triangle $(v,w,u)$ at a later step.
The above papers avoid this unwanted redundancy by using an ordering, explicitly or not.
We use the framework developed in~\cite{ortmann2014}:
a total ordering $\pi$ is defined over the vertices, and the triple $(u,v,w)$ is only considered a valid triangle if $\pi_u<\pi_v<\pi_w$.
%
This guarantees that each triangle is listed only once:
as illustrated in Figure~\ref{fig:skeleton-triangle}, vertices in any triangle of the DAG $G_\pi$
appear in one and only one of 3 positions:
$u$ is first, $v$ is second, $w$ is third;
the same holds for edges: $L$ is the long edge, and $S_1$ and $S_2$ are the first and second short edges.
%
It leads to 3 variants of adjacency testing (seed vertex $v$ or $w$ instead of $u$) and of neighborhood intersection (seed edge $L$ or $S_2$ instead of $S_1$). 

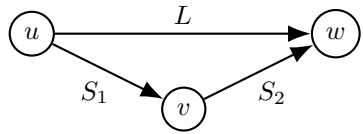
\begin{figure}[t]
  \begin{center} \begin{tikzpicture}[scale=1, transform shape,mystyle/.style={fill=green!60!black!100}]
  \def\arr{{Latex[length=3mm]}}
  \node [circle, thick, draw] (t1) at (-2,0) {$u$};
  \node [circle, thick, draw] (t2) at (0,-1) {$v$};  
  \node [circle, thick, draw] (t3) at (2,0) {$w$};
  \draw [-\arr, thick, below] (t1) edge node{$S_1\quad$} (t2);
  \draw [-\arr, thick, above] (t1) edge node{$L$} (t3);
  \draw [-\arr, thick, below] (t2) edge node{$\quad S_2$} (t3);
  \end{tikzpicture} \end{center}  
  \caption{\textbf{Directed triangle} with the unified notations proposed in~\protect\cite{ortmann2014}. The edges are directed according to an ordering $\pi$ such that $\pi_u<\pi_v<\pi_w$.}\label{fig:skeleton-triangle}
\end{figure}

Choosing the right data-structure is key to the performance of algorithms.
All triangle listing algorithms have to visit the neighborhoods of vertices.
Using hash table or binary tree to store them is very effective: they respectively allow for constant and logarithmic search on average.
However, because of high constants, they are reportedly slow in terms of actual running time~\cite{schank2005finding}.
A faster structure is the boolean array used in {\tt K3} for neighborhood intersection.
It registers the elements of $N^+_u$ in a boolean table $B$ so that, for each neighbor $v$ of $u$, it is possible to check in constant time if a neighbor $w$ of $v$ is also a neighbor of $u$.
This is the structure used by the fastest methods~\cite{ortmann2014,danisch2018listing}.

\begin{figure*}[!ht]
\begin{center}
\begin{minipage}{0.49\textwidth}
\begin{algorithm}[H]
    \centering
    \caption{-- \APP{} (or \texttt{L+n})}\label{algo:PP}
    \begin{algorithmic}[1]
     \For {each vertex $v$} $B[v]\leftarrow$ False\EndFor
     \For {each vertex $w$}
      \For {$v\in N^{-}_w$} $B[v]\leftarrow$ True \EndFor
      \For {$u\in N^{-}_w$} \label{alg:PP:foru}
        \For {$v\in N^{+}_u$}
            \If{$B[v]$}
                \Output triangle$\{u,v,w\}$
            \EndIf
        \EndFor
      \EndFor
      \For {$v\in N^{-}_w$} $B[v]\leftarrow$ False \EndFor
     \EndFor
    \end{algorithmic}
    \hrule\vspace{2mm}
    Complexity:\vspace{-3mm}
    $$ \Theta \Big( m + \sum_{(u,w)\in E_\pi} d^+_u \Big) = \Theta \Big(\sum _{u \in V} {d^+_u}^2 \Big) $$
    \vspace{-3mm}
\end{algorithm}
\end{minipage}
\hfill
\begin{minipage}{0.49\textwidth}
\begin{algorithm}[H]
    \centering
    \caption{-- \APM{} (or \texttt{S$_1$+n})}\label{algo:PM}
    \begin{algorithmic}[1]
     \For {each vertex $w$} $B[w]\leftarrow$ False\EndFor
     \For {each vertex $u$}
      \For {$w\in N^{+}_u$} $B[w]\leftarrow$ True \EndFor
      \For {$v\in N^{+}_u$} \label{alg:PM:forv}
        \For {$w\in N^{+}_v$}
            \If{$B[w]$}
                \Output triangle$\{u,v,w\}$
            \EndIf
        \EndFor
      \EndFor
      \For {$w\in N^{+}_u$} $B[w]\leftarrow$ False \EndFor
     \EndFor
    \end{algorithmic}
    \hrule\vspace{2mm}
    Complexity:\vspace{-3mm}
    $$ \Theta \Big( m + \sum_{(u,v)\in E_\pi} d^+_v \Big) = \Theta \Big( m + \sum _{v \in V} d^+_v d^-_v \Big) $$
    \vspace{-3mm}
\end{algorithm}
\end{minipage}
\end{center}
\vspace{-5mm}
\end{figure*}

In the rest of this paper, we therefore only consider triangle listing algorithms that use neighborhood intersection and a boolean array.
We present the two that we will study
in Algorithms~\ref{algo:PP} and~\ref{algo:PM} with the notations of Figure~\ref{fig:skeleton-triangle} for the vertices.
They initialize the boolean array $B$ to false (line 1), consider a first vertex (line 2) and store its neighbors in $B$ (line 3); then, for each of its neighbors (line 4), they check if their neighbors (line 5) are in $B$ (line 6), in which case the three vertices form a triangle (line 7). $B$ is reset (line 8) before continuing with the next vertex.
The Algorithm~\ref{algo:PP} corresponds to \texttt{L+n} in~\cite{ortmann2014};
we call it \APP{} \flok{because of the two \ltok{``+" (referring to out-degrees)} involved in its complexity.}
%
The Algorithm~\ref{algo:PM} corresponds to \texttt{S$_1$+n} in~\cite{ortmann2014};
we call it \APM{}\,%
\footnote{A third natural variant exists: \AMM{} or \texttt{S$_2$+n}. We ignore it here since its complexity is equivalent to the one of \APP{}.}.
Their complexities are given in Property~\ref{complexity-apmp}. Since they depend on the indegree and outdegree of vertices, the choice of ordering will impact the running time of the algorithms.

\begin{property}[Complexity of \APP{} and \APM{}]
\label{complexity-apmp}
The time complexity of \APP{} is $\Theta (\sum _{u \in V} {d^+_u}^2 )$.
The time complexity of \APM{} is $\Theta \left( m + \sum _{v \in V} d^+_v d^-_v \right)$. 
\end{property}

\begin{proof}
In both algorithms, the boolean table $B$ requires $n$ initial values, $m$ set and $m$ reset operations, which is $\Theta(m)$ assuming that $n\in {\cal O}(m)$.
In \APP{}, a given vertex $u$ appears in the loop of line~\ref{alg:PP:foru} as many times as it has a successor $w$; every time, a loop over each of its successors $v$ is performed. 
In total, $u$ is involved in $\Theta({d^+_u}^2)$ operations.
Similarly, in \APM{}, a given vertex $v$ appears in the loop of line~\ref{alg:PM:forv} as many times as it has a predecessor $u$; every time, a loop over each of its successors $w$ is performed.
In total, $v$ is involved in $\Theta(d^+_v d^-_v)$ operations.
The term $m$ is omitted in the complexity of \APP{} as $\sum_{u \in V} {d^+_u}^2 \geq \sum _{u \in V} d^+_u=m$, but not in \APM{} as $\sum _{v \in V} d^+_v d^-_v$ can be lower than $m$.
\end{proof}

\subsection{Orderings and complexity bounds.}
Ortmann and Brandes~\cite{ortmann2014} order the vertices by non-decreasing degree or core value.
In their experimental comparison, they test several algorithms as well as \APP{} and \APM{}, each with degree ordering, core ordering, and with the original ordering of the dataset.
They conclude that the fastest method is 
\APP{} with core or degree ordering:
core is faster to list triangles when the ordering is given as an input, and degree is faster when the time to compute the ordering is also included.

Danisch \textit{et al.}~\cite{danisch2018listing} also use core
ordering in the more general problem of listing $k$-cliques. For
triangles ($k=3$), their algorithm is equivalent to \APM{}, and they show that using core ordering outperforms the methods
of~\cite{chiba1985arboricity,latapy2008main,kanade2004cliques}.

With these two orderings, it is possible to obtain upper-bounds for the time complexity in terms of graph properties.
Chiba and Nishizeki~\cite{chiba1985arboricity} show that \texttt{K3} with degree ordering has a complexity in $ \mathcal{O} (m \cdot \alpha(G))$, where $\alpha(G)$ is the arboricity of graph $G$. 
With core ordering, \texttt{node-iterator-core}~\cite{schank2005finding} and \texttt{kClist}~\cite{danisch2018listing} have complexity $\mathcal{O} (m \cdot c(G))$, where $c(G)$ is the core value of graph $G$.  
These bounds are considered equal in~\cite{ortmann2014}, following the proof in~\cite{zhou1999edge} that $\alpha(G) \leq c(G) \leq 2 \alpha(G)-1$.
However, we focus in this work on the complexities expressed in Algorithms~\APP{} and~\APM{} as we will see that they describe the running time more accurately.






\section{New orderings to reduce the cost of~triangle~listing}\label{method}


\subsection{Formalizing the cost of triangle listing algorithms.}

In this section, we discuss how to design vertex orderings to reduce the cost of triangle listing algorithms.
For this purpose, we introduce the following costs that appear in the complexity formulas of Algorithms~\ref{algo:PP} and~\ref{algo:PM}.
Recall that the initial graph is undirected and that the orientation of the edges is given by the ordering $\pi$, which partitions neighbors into successors and predecessors.

\begin{definition}[Cost induced by an ordering]
Given an undirected graph~$G$, the costs \DPP{} and \DPM{}  induced by a vertex ordering $\pi$ are defined by:
$$
\DPP (\pi)=\sum_{u\in V} d_u^+d_u^+
\qquad\qquad 
\DPM (\pi)=\sum_{u\in V} d_u^+d_u^-
$$
\end{definition}

The fastest methods in the state of the art are \APP{} with core or degree ordering~\cite{ortmann2014},
and \APM{} with core ordering~\cite{danisch2018listing}.
The intuition of both orderings is that high degree vertices are ranked after most of their neighbors in $\pi$ so that their outdegree in  $G_\pi$ is lower.
This reduces the cost \DPP{}, which in turn reduces the number of operations required to list all the triangles as well as the actual running time of \APP{}. In~\cite{ortmann2014}, it is mentioned that core ordering performs well with \APM{} as a side effect.


To our knowledge, no previous work has designed orderings with a low \DPM{} cost and used them with \APM{}.
We will show that such orderings can lower the computational cost further.
Yet, optimizing \DPM{} or \DPP{} is computationally hard because of Theorem~\ref{thm:dpmp}:

\begin{theorem}[NP-hardness]\label{thm:dpmp}
Given a graph $G=(V,E)$, it is NP-hard to find an ordering~$\pi$ on $V$ that minimizes $\text{\DPM}(\pi)$ or that minimizes $\text{\DPP}(\pi)$.
\end{theorem}

\begin{proof}
For the hardness of \DPM{}, a proof was already known from~\cite{stack2017dpm} but never published as far as we know; 
we give a new and simpler proof in Appendix~\ref{appendix:dpm}.
We prove the result for \DPP{} in Appendix~\ref{appendix:dpp}.
\end{proof}

\subsection{Distinguishing two tasks for triangle listing.}

Triangle listing typically consists of the following steps: loading a graph, computing a vertex ordering, and listing the triangles.
Time measurements
in~\cite{latapy2008main,danisch2018listing,li2020cliques} only take
the last step into account, while \cite{schank2005finding,ortmann2014}
also include the other steps. We therefore address two distinct tasks in our study: 
we call {\bf \listonly{}} the task of listing the triangles of an already loaded graph with a given vertex ordering; 
we call {\bf \listfull{}} the task of loading a graph, computing a vertex ordering, and listing its triangles.

In the rest of the paper, we use the notation \textit{task-order-algorithm}: for instance, \adjonly{}-core-\APM{} refers to the \listonly{} task with core ordering and algorithm \APM{}.
Using this notation,
the fastest methods identified in the literature are \adjonly{}-core-\APM{} in~\cite{danisch2018listing}, \adjonly{}-core-\APP{} and \adjfull{}-degree-\APP{} in~\cite{ortmann2014}. We use all three methods as benchmarks in our experiments of Section~\ref{experiments}.

With \listonly{}, the ordering time is not taken into account, which allows to spend a long time to find an ordering with low cost. On the other hand, \listfull{} 
favors quickly obtained orderings even if their induced cost is not the lowest. For this reason, there is a time-quality trade-off for cost-reducing heuristics.


\subsection{Reducing \DPM{} along a time-quality trade-off.} \label{dpm-orderings}

We remind that two efficient algorithms are identified in the literature for triangle listing
%
%
(see Algorithms~\ref{algo:PP} and~\ref{algo:PM}). Their number of operations are respectively \DPP{} and \DPM{}. However, the orderings that have been considered (degree and core) induce a low \DPP{} cost, but not necessarily a low \DPM{} cost.

Our goal here is therefore to design a
procedure that takes a graph as input and produces an ordering $\pi$
with a low induced cost $\DPM(\pi)$. Because of Theorem~\ref{thm:dpmp}, finding an optimal solution 
is not realistic
for graphs with millions of edges.
We therefore present three heuristics 
aiming at reducing the \DPM{} value,
exploring the trade-off between quality in terms of \DPM{} and ordering time. 

\subsubsection{\textit{Neigh} heuristic.} 

%


\begin{algorithm}[b]
\caption{Neighborhood optimization (\textit{Neigh} heuristic)}
\label{algo:dpmopt}
\renewcommand{\algorithmicrequire}{\textbf{Input:}}
\renewcommand{\algorithmicuntil}{\textbf{while}}
\begin{algorithmic}[1]
\Require graph $G$, initial ordering $\pi$,  threshold $\epsilon\geq 0$
\Repeat
\State $C_0=\DPM{}(\pi)$
\For {each vertex $u$ of $G$}\label{alg:neigh:foru}
    \State sort $N_u$ according to $\pi$ \label{alg:neigh:sort}
    \State $p_*=\textrm{argmin}_{p \in \llbracket 0,d_u\rrbracket }\ \{ \DPM{}(p) \}$\label{alg:neigh:pbest}
    \State update ordering $\pi$ to put $u$ in position $p_*$\label{alg:neigh:update}
 \EndFor
\Until{$\DPM{}(\pi) < (1-\epsilon) \cdot C_0$}\label{alg:neigh:until}
\end{algorithmic}
\end{algorithm}


%
We define the \textit{neighborhood  optimization} method, a greedy reordering where each vertex is placed at the optimal index with respect to its neighbors, \ltok{as illustrated in Figure~\ref{fig:neigh_heur}}.
First, notice that changing an index $\pi_u$ only affects $\DPM{}(\pi)$ if the position of $u$ with respect to at least one of its neighbors changes;
otherwise the in- and outdegrees of all vertices remain unchanged.
Starting from any ordering $\pi$,  the algorithm described in Algorithm~\ref{algo:dpmopt} considers each vertex $u$ one by one (line~\ref{alg:neigh:foru}) and, for each
$p\in \llbracket 1,d_u\rrbracket$, it computes $\DPM(p)$, the
value of $\DPM$ when $u$ is just after its $p$-th neighbor in $\pi$, as well as $\DPM(0)$ when $u$ is before all its neighbors.
The position $p_*$ that induces the lowest value of $\DPM$ is
selected (line~\ref{alg:neigh:pbest}) and the ordering is updated (line~\ref{alg:neigh:update}). The process is repeated until $\DPM$ reaches a local minimum, or
until the relative improvement is under a threshold $\epsilon$ (last
line). The resulting $\pi$ induces a low $\DPM$ cost.
%


\begin{figure}[t]

\begin{center}
    
\begin{tikzpicture}[thick,scale=0.5] 


\node[shape=circle,minimum size=6mm,draw=black,ultra thick] (A) at (2,7) {$a$};
\node[shape=circle,minimum size=6mm,draw=black] (B) at (4.5,6) {$b$};
\node[shape=circle,minimum size=6mm,draw=black] (C) at (7,7) {$c$};
\node[shape=circle,minimum size=6mm,draw=black] (D) at (-0.5,5) {$d$};
\node[shape=circle,minimum size=6mm,draw=black] (E) at (2,3) {$e$};
\node[shape=circle,minimum size=6mm,draw=black,scale=0.9] (F) at (4.5,4) {$f$} ;
\node[shape=circle,minimum size=6mm,draw=black] (G) at (7,3) {$g$} ;
\path [-latex](A) edge (D);
\path [-latex](A) edge (B);
\path [-latex](A) edge (E);
\path [-latex](B) edge (C);
\path [-latex](B) edge (E);
\path [-latex](B) edge (G);
\path [-latex](C) edge (G);
\path [-latex](D) edge (E);
\path [-latex](E) edge (F);
\path [-latex](F) edge (G);
   

\draw (0,0) -- (7,0) -- (7,1) -- (0,1) -- (0,0);
\draw[ultra thick] (0,0) -- (0,1) -- (1,1) -- (1,0) -- (-0.05,0);
\draw (2,0) -- (2,1);
\draw (3,0) -- (3,1);
\draw (4,0) -- (4,1);
\draw (5,0) -- (5,1);
\draw (6,0) -- (6,1);
\node[text width=1cm, font=\normalsize] at (1.33,0.4) {$a$};   
\node[text width=1cm, font=\normalsize] at (2.33,0.49) {$b$};   
\node[text width=1cm, font=\normalsize] at (3.33,0.4) {$c$};   
\node[text width=1cm, font=\normalsize] at (4.33,0.49) {$d$}; 
\node[text width=1cm, font=\normalsize] at (5.33,0.4) {$e$};   
\node[text width=1cm, font=\normalsize] at (6.33,0.42) {$f$};   
\node[text width=1cm, font=\normalsize] at (7.33,0.32) {$g$};   

\path [-stealth, very thick] (0.5,1) edge [bend left=45] (5,1);

\end{tikzpicture}
\end{center}

\hspace{2cm}
\begin{center}
\begin{tikzpicture}[thick,scale=0.5] 


\node[shape=circle,minimum size=6mm,draw=black,ultra thick] (A) at (2,7) {$a$};
\node[shape=circle,minimum size=6mm,draw=black] (B) at (4.5,6) {$b$};
\node[shape=circle,minimum size=6mm,draw=black] (C) at (7,7) {$c$};
\node[shape=circle,minimum size=6mm,draw=black] (D) at (-0.5,5) {$d$};
\node[shape=circle,minimum size=6mm,draw=black] (E) at (2,3) {$e$};
\node[shape=circle,minimum size=6mm,draw=black,scale=0.9] (F) at (4.5,4) {$f$} ;
\node[shape=circle,minimum size=6mm,draw=black] (G) at (7,3) {$g$} ;
\path [-latex](D) edge (A);
\path [-latex](B) edge (A);
\path [-latex](E) edge (A);
\path [-latex](B) edge (C);
\path [-latex](B) edge (E);
\path [-latex](B) edge (G);
\path [-latex](C) edge (G);
\path [-latex](D) edge (E);
\path [-latex](E) edge (F);
\path [-latex](F) edge (G);


\draw (0,0) -- (7,0) -- (7,1) -- (0,1) -- (0,0);
\draw[ultra thick] (4,0) -- (4,1) -- (5,1) -- (5,0) -- (4-0.05,0);
\draw (1,0) -- (1,1);
\draw (2,0) -- (2,1);
\draw (3,0) -- (3,1);
\draw (4,0) -- (4,1);
\draw (5,0) -- (5,1);
\draw (6,0) -- (6,1);
\node[text width=1cm, font=\normalsize] at (1.33,0.49) {$b$};   
\node[text width=1cm, font=\normalsize] at (2.33,0.4) {$c$};   
\node[text width=1cm, font=\normalsize] at (3.33,0.49) {$d$};
\node[text width=1cm, font=\normalsize] at (4.33,0.4) {$e$};   
\node[text width=1cm, font=\normalsize] at (5.33,0.4) {$a$};   
\node[text width=1cm, font=\normalsize] at (6.33,0.42) {$f$};   
\node[text width=1cm, font=\normalsize] at (7.33,0.32) {$g$};  

\end{tikzpicture}
    
\end{center}
    \caption{\textbf{Example of update in the \textit{Neigh} heuristic}: vertex $a$ is moved to a  position among its neighbors that induces the lowest cost. The tables indicates how the ordering is updated. The edge in the DAG are reoriented accordingly. Here, the ordering at the top has $\DPM{} = 9$ while the ordering at the bottom has $\DPM{} = 6$. For this graph, the optimal \DPM{} cost is 3 (with ordering $e,g,f,a,c,d,b$). 
    }
    \label{fig:neigh_heur}
\end{figure}
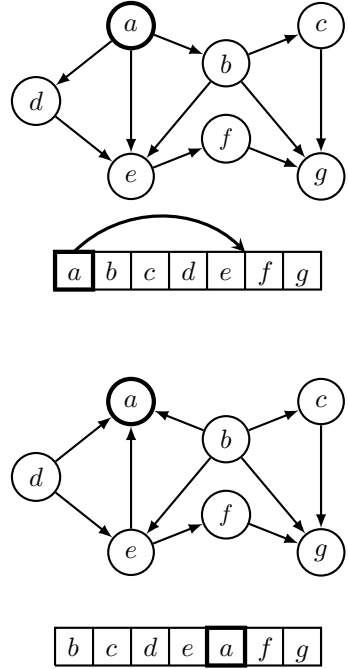


For a vertex $u$, sorting the neighborhood according to $\pi$ takes
$\mathcal{O}(d_u\log d_u)$ operations; finding the best position
takes $\Theta(d_u)$ because it only depends on the values $d_v^+$ and $d_v^-$ of each neighbor $v$ of $u$. With a linked list, $\pi$ is updated in constant time. If $\Delta$ is the highest degree in the graph, one iteration over all the vertices thus takes
$\mathcal{O}(m\log \Delta)$, which leads to a total complexity 
$\mathcal{O}(Im\log \Delta)$ if the improvement threshold $\epsilon$ is reached after
$I$ iterations. Notice that on all the tested datasets the process reaches $\epsilon=10^{-2}$ after less than ten iterations.

This heuristic has several strong points:
it can be used for other objective functions, for instance \DPP{};
it is greedy, so the cost keeps improving until the process stops;
if the initial ordering already induces a low \DPM{} cost, the heuristic can only improve it; 
it is stable in practice, which means that starting from several random orderings give similar final costs;
and we show in Section~\ref{experiments} that it allows for the fastest \listonly{}.

In spite of its log-linear complexity, this heuristic can take longer than the actual task of listing triangles in practice, which is an issue for the \listfull{} task.
We therefore propose the following faster heuristics in the case of the \listfull{} task.

\subsubsection{\textit{Check} heuristic.}
This heuristic is inspired by core ordering, where vertices are repeatedly selected according to their current degree~\cite{batagelj2003kcore}.
It considers all vertices by decreasing degree and
\textit{checks} whether it is better to put a vertex at the beginning or at the end of the ordering.
More specifically, $\pi$ is obtained as follows: before placing vertex $u$, let $V_b$ (resp. $V_e$) be the vertices that have been placed at the beginning (resp. at the end) of the ordering, and $V_?$ those that are yet to place.
The neighbors of $u$ are partitioned in
$N_b=N_u\cap V_b$, $N_e=N_u\cap V_e$ and $N_?=N_u\cap V_?$.
We consider two options to place $u$: either just after the vertices in $V_b$ ($\pi_u=|V_b|+1$), or
just before the vertices in $V_e$ ($\pi_u=n-|V_e|$).
In either case, $u$ has all vertices of $N_b$ as predecessors, and all vertices of $N_e$ as successors.
In the first case, vertices in $N_?$ become successors, which induces a \DPM{} cost $C_b=|N_b|\cdot (|N_e|+|N_?|)$.
In the second, the cost is $C_e=(|N_b|+|N_?|)\cdot |N_e|$.
The option with the smaller cost is selected.
Sorting the vertices by degree requires $\mathcal{O}(n)$ steps with bucket sort.
Maintaining the sizes of $N_b$, $N_e$, $N_?$ for each vertex requires one update for each edge.
Therefore, the complexity is $\mathcal{O}(m+n)$, or $\mathcal{O}(m)$ assuming that $n \in {\cal O}(m)$.

\subsubsection{\textit{Split} heuristic.}
Finally, we propose a heuristic that is faster to achieve but compromises on the quality of the resulting ordering.
Degree ordering has been identified as the best solution for \listonly{} with algorithm-\APP{}~\cite{ortmann2014}. We adapt it for \DPM{} by \textit{splitting} vertices 
alternatively at the beginning and at the end of the ordering $\pi$.
More precisely, a non-increasing degree ordering $\delta$ is computed, then the vertices are split according to their parity: if $u$ has index $\delta_u=2i+1$ then $\pi_u=i+1$; if $\delta_u=2i$, then $\pi_u=n+1-i$.
Thus, high degree vertices will have either few predecessors or few successors, which ensures a low \DPM{} cost.
With the graph of Figure~\ref{fig:neigh_heur}, supposing that we start from the non-decreasing degree ordering $(e,b,g,a,f,d,c)$, which has $\DPM{}=7$,
the \textit{Split} method leads to $(e,g,f,c,d,a,b)$, which has $\DPM{}=4$.
%
The complexity of this method is in $\mathcal{O}(n)$ like the degree ordering.


\section{Experiments}\label{experiments}
\subsection{Experimental setup.}\label{exp:setup}
\subsubsection{Datasets.}\label{datasets} We use the 12 real-world graphs
described in Table~\ref{table:datasets}.
Loops have been removed and the directed graphs have been
transformed into undirected graphs by keeping one edge when one existed
in either or both directions.

\begin{table}[h!]
\newcommand{\dataWeb}{$\bigstar$}
\newcommand{\dataSoc}{$\blacktriangle$}
\newcommand{\dataCit}{$\blacksquare$}
\caption{\textbf{Datasets used for the experiments}, ranked by number of edges. They represent either web networks \dataWeb{}, social networks \dataSoc{} or citation networks \dataCit{}.}\label{table:datasets}

\centering
\begin{scriptsize}
\begin{tabular}{r|c|c|c}
	dataset [source] & vertices & edges & triangles\\
	\hline
	skitter \dataWeb \cite{dataset-snap} &	1,696,415	&	11,095,298	&	28,769,868\\
	patents \dataCit \cite{dataset-snap} &	3,774,768	&	16,518,947	&	7,515,023\\
	baidu \dataWeb \cite{dataset-networkRepository} &	2,141,301	&	17,014,946	&	25,207,196\\ 
	pokec \dataSoc \cite{dataset-snap} &	1,632,804	&	22,301,964	&	32,557,458\\
	socfba \dataSoc \cite{dataset-networkRepository} &	3,097,166	&	23,667,394	&	55,606,428\\
	LJ \dataSoc \cite{dataset-snap} &	4,036,538	&	34,681,189	&	177,820,130\\
	wiki \dataWeb \cite{dataset-snap} &	2,070,486	&	42,336,692	&	145,707,846\\
	orkut \dataSoc \cite{dataset-snap} &	3,072,627	&	117,185,083	&	627,584,181\\
	it \dataWeb \cite{boldi2004webgraph} &	41,291,318	&	1,027,474,947	&	48,374,551,054\\
	twitter \dataSoc \cite{boldi2004webgraph} &	41,652,230	&	1,202,513,046	&	34,824,916,864\\
	friendster \dataSoc \cite{dataset-snap} &	124,836,180	&	1,806,067,135	&	4,173,724,142\\
	sk \dataWeb \cite{boldi2004webgraph} &	50,636,151	&	1,810,063,330	&	84,907,041,475\\
\end{tabular}
\end{scriptsize}
\end{table}

\subsubsection{Software and hardware.}
We release a uniform open-source implementation\,\footnote{\url{\codegit}}
of \APP{} and \APM{} algorithms, as well as the different ordering strategies that we discussed in Section~\ref{method}.
Our implementation allows to run either algorithm in parallel, which is possible because each iteration of the main loop is independent from the others.
Among orderings however, only degree and \textit{Split} are easily parallelizable; to be consistent, we use a single thread to compare the different methods.
%
%
The code is in \texttt{c++} and uses \texttt{gnu make 4} and the compiler \texttt{g++ 8.2}
with optimization flag \texttt{Ofast} and \texttt{openmp} for parallelisation.
We run all the programs on
a \texttt{sgi ub2000 intel xeon e5-4650L @2.6 GHz, 128Gb ram} running \texttt{linux suse 12.3}.

Regarding the state of the art, 
%
%
the most competitive implementation available for triangle listing is \texttt{kClist} in {\tt c}~\cite{danisch2018listing}, which has already been shown to outperform previous programs~\cite{kanade2004cliques,latapy2008main}. 
It lists $k$-cliques using a core ordering and a recursive algorithm that is equivalent to \APM{} for $k=3$. 
We compared our implementation to \texttt{kClist} in various settings and found that ours is 14\% faster on average, presumably because it does not use recursion.
Moreover, the paper that identified core-\APP{} and 
degree-\APP{} as the fastest methods~\cite{ortmann2014} does not provide the corresponding code.
Therefore, we only use our own implementation of \APM{} and \APP{} in the rest of this paper: we exclusively focus on the speedup caused by the vertex ordering, separating it from the speedup originating from the implementation.
%


\subsection{Cost and running time are linearly correlated.}
In order to show that the cost functions \DPP{} and \DPM{} are good estimates of the running time,
we measure the correlation between the running time of \listonly{}
 and the corresponding cost induced by various orderings (core, degree, our heuristics, but also breadth- and depth-first search, random ordering, etc).
%
In~Figure~\ref{img:time-wrt-dpm}, we see that the running time for a given dataset correlates almost linearly to the corresponding cost:
the lines represent linear regressions.
It only presents some of the datasets for readability, but the correlation is above 0.82 on all of them.
%
In other words, the execution time of a listing algorithm is almost a linear function of the cost induced by the ordering, which is why reducing this cost actually improves the running time, as we will see.

\begin{figure}[h]
    \centering
    \includegraphics[width=\linewidth]{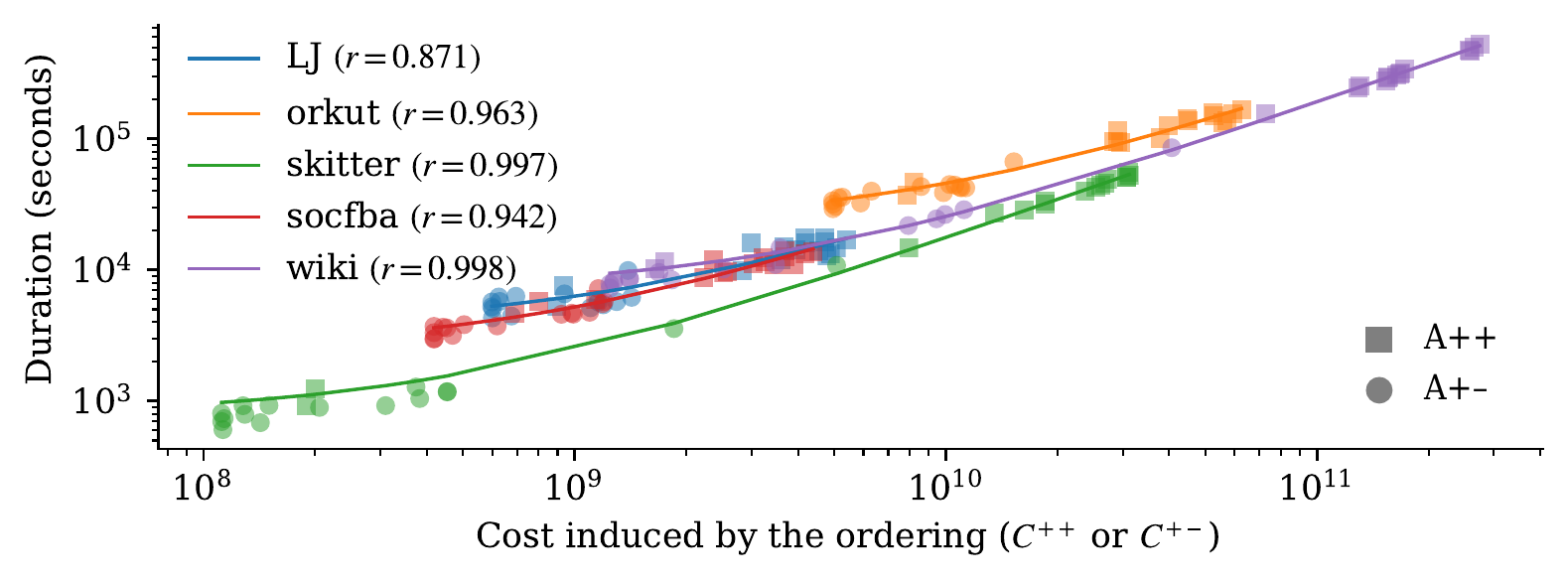}
    \caption{\textbf{Algorithm running time vs the cost induced by the ordering.}
      Each mark represents an ordering: circles are for cost \DPM{} and algorithm \APM{}, squares are for cost \DPP{} and algorithm \APP{}. 
       Each color represents a dataset: the line of linear regressions and associated correlation coefficients $r$ show the proportionality between cost and time. }
\label{img:time-wrt-dpm}
\end{figure}

\begin{figure*}[t]
    \centering
    \includegraphics[width=.98\textwidth]{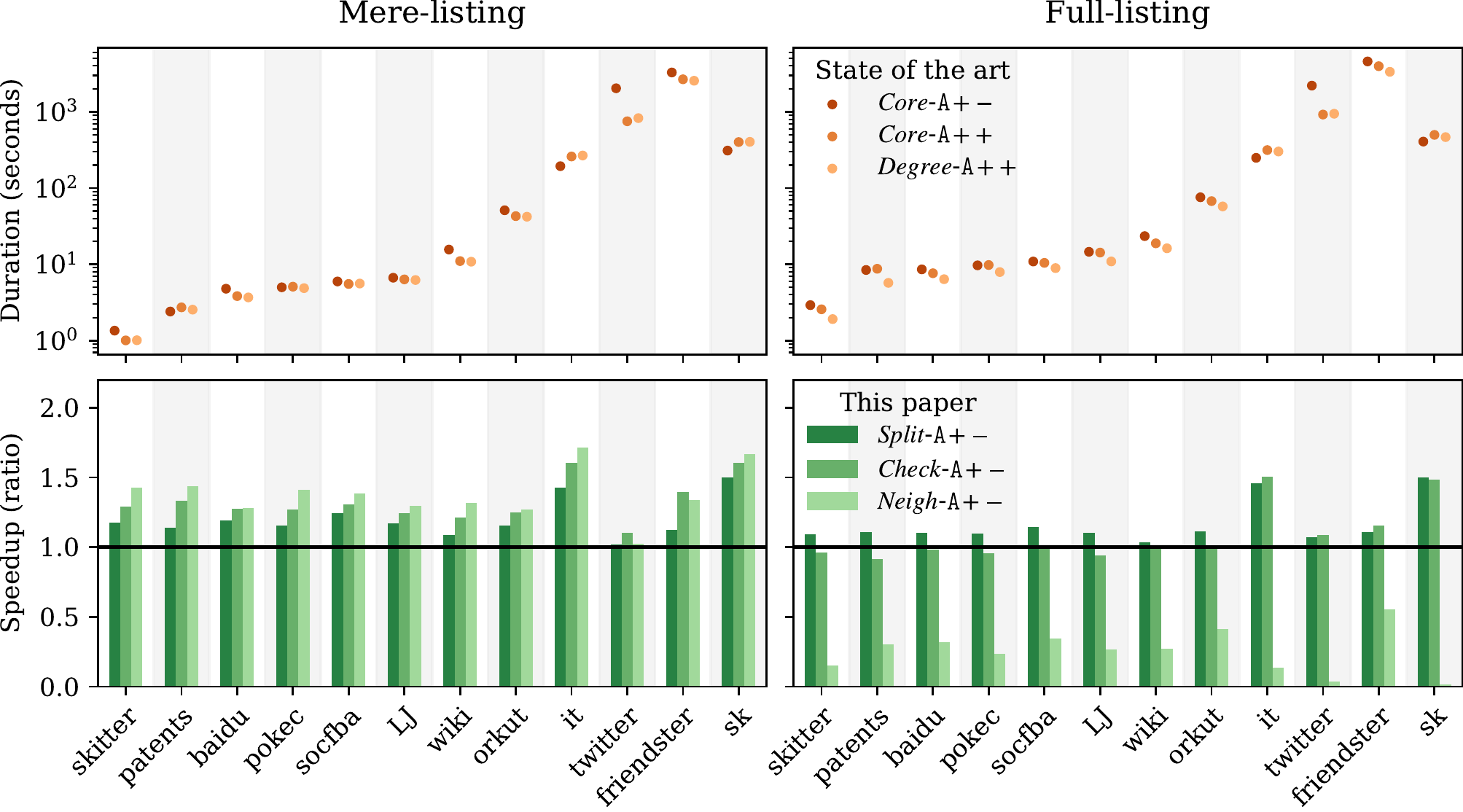}
\caption{\textbf{Comparison of state-of-the-art methods and speedup of our methods.} The top charts show the runtime of the three state-of-the-art methods; depending on the dataset, the fastest method is not always the same. The bottom charts show the speedup of our three methods against the fastest existing method of each dataset. On the left, for \listonly{}, we see that our three heuristics consistently outperform the three state-of-the-art methods, and that \textit{Neigh} or \textit{Check} are the fastest. On the right, for \listfull{}, \textit{Neigh} is not efficient but \textit{Split} is always faster than existing methods and \textit{Check} is faster on bigger datasets.}
\label{img:listing-combined}
\end{figure*}

\subsection{\textit{Neigh} outperforms previous \listonly{} methods.}

We compare our methods to the state of the art for \listonly{} (core-\APM{} in~\cite{danisch2018listing} and core-\APP{} in~\cite{ortmann2014}) and for \listfull{} (degree-\APP{} in~\cite{ortmann2014}) in Figure~\ref{img:listing-combined}.
The top charts present the running time of the three state-of-the-art methods for all datasets, for the \listonly{} task (left)
and the \listfull{} task (right). 
We can see that there is no clear winner for \listonly{}: 
both \APP{} methods have a very similar duration, but core-\APM{} can be between 1.4 times faster and 2.4 times slower depending on the dataset.
This explains why \cite{ortmann2014} and \cite{danisch2018listing} did not agree on the fastest method.  

On the other hand, our heuristics \textit{Neigh}, \textit{Check} and \textit{Split}
manage to produce orderings significantly lower \DPM{} costs. This translates directly into short running times for \listonly{} with \APM{}.
To compare our contributions with the state of the art, we take for each dataset the fastest of the three existing methods.
The bottom left chart of Figure~\ref{img:listing-combined} shows the speedup of our methods compared to the fastest existing one. Exact runtimes of the best existing and of the methods proposed in this work are reported in Table~\ref{tab:runtimes}.


The main result is that \textit{Neigh}-\APM{} is always faster than the best previous method. 
The speedup is 1.38 on average and ranges from only 1.02 on \textit{twitter} to 1.71 on the \textit{it} dataset.
\textit{Check}-\APM{} is almost as good, with a 1.32 average speedup ranging from 1.10 to 1.60; it is even faster than \textit{Neigh}-\APM{} on two of the datasets.
\textit{Split}-\APM{} is a little slower, which is expected because this ordering is designed to be obtained quickly and does not reduce \DPM{} as efficiently as our other heuristics.
However it still consistently outperforms all the previous methods, with a 1.20 average speedup.



%

\begin{table}[!h]
    \centering
\begin{footnotesize}
    \begin{tabular}{r||c|c||c|c}
    & \multicolumn{2}{c||}{\textbf{\listonly{}}}& \multicolumn{2}{c}{\textbf{\listfull{}}} \\
    dataset & existing & this paper & existing & this paper \\ 
    \hline
skitter & 1.00s & 0.71s & 1.91s & 1.75s \\
patents & 2.40s & 1.67s & 5.71s & 5.15s \\
baidu & 3.68s & 2.87s & 6.38s & 5.77s \\
pokec & 4.87s & 3.44s & 7.91s & 7.21s \\
socfba & 5.52s & 3.98s & 8.92s & 7.79s \\
LJ & 6.23s & 4.79s & 10.91s & 9.88s \\
wiki & 10.82s & 8.22s & 16.23s & 15.65s \\
orkut & 42.11s & 33.09s & 57.47s & 51.60s \\
it & 3m13 & 1m53 & 4m09 & 2m45 \\
twitter & 12m31 & 11m20 & 15m21 & 14m08 \\
friendster & 42m36 & 30m31 & 55m47 & 48m13 \\
sk & 5m10 & 3m06 & 6m47 & 4m31 \\
    \end{tabular}
\end{footnotesize}
    \caption{\textbf{Duration of triangle listing of existing methods against methods of this paper.} For each dataset, we compare the fastest state-of-the-art method against the fastest of our methods. Recall that \listonly{} only takes into account the runtime of the listing algorithm (\APP{} or \APM{}) while \listfull{} also counts the graph loading time and the ordering time.}
    \label{tab:runtimes}
\end{table}

\subsection{\textit{Split} outperforms previous \listfull{} methods.}

For \listfull{}, the top right chart of Figure~\ref{img:listing-combined} compares the three state-of-the-art methods and shows that degree-\APP{} is the fastest for almost all datasets. 
This result is consistent with the result reported in~\cite{ortmann2014}, that specifically addresses \listfull{}.
The bottom right chart shows the speedup of our three methods compared to the fastest state-of-the-art method. 
Note that the \textit{Neigh} heuristic is not competitive here (speedup under one) since its ordering time is long compared to other methods. 

The main result is that \textit{Split}-\APM{} is always faster than previous methods.
The speedup compared to existing methods is 1.16 on average, and it ranges from 1.04 on \textit{wiki} to 1.50 on \textit{it} dataset.
%
\textit{Check} also gives very good results: on medium datasets, it is a bit slower than degree-\APP{}, but it outperforms all state-of-the-art methods on large datasets (\textit{it, twitter, friendster, sk}), and it even beats \textit{Split} on three of them. 
This hints at a transition effect: 
the \textit{Check} ordering has a lower \DPM{} value but it takes $\mathcal{O}(m)$ steps to compute, 
while \textit{Split} only needs $\mathcal{O}(n)$;
for larger datasets, the listing step prevails, so the extra time spent to compute \textit{Check} becomes profitable.





\section*{Conclusion}

In this work, we address the issue of in-memory triangle listing in large graphs.
We formulate explicitly the computational costs of the most efficient existing algorithms, and investigate how to order vertices to minimize these costs.
After proving that the optimization problems are NP-hard, we propose scalable heuristics that are specifically tailored to reduce the costs induced by the orderings.
We show experimentally that these methods outperform
the current state of the art for both the \listonly{} and the \listfull{} tasks.

Our results also emphasize a limitation in the possible acceleration: while it is certainly possible to keep improving the \listonly{} step, a significant part of \listfull{} is spent on other steps: computing the ordering, but also loading the graph or writing the output. 
It seems, however, that the \listonly{} step takes more importance as graphs grow larger, which makes our listing methods all the more relevant for future, larger datasets.
A natural extension of this work is to use similar vertex ordering heuristics in the more general case of 
clique listing.
Formulating appropriate cost functions for clique listing algorithms is not straightforward and requires studying precisely the different possibilities to detect all the vertices of a clique.

\section*{Acknowledgements}
We express our heartfelt thanks to Maximilien Danisch who initiated this project. We also thank Alexis Baudin, Esteban Bautista, Katherine Byrne and Matthieu Latapy for their valuable comments.
This work is funded by the ANR (French National Agency of Research) partly by Limass project (under grant ANR-19-CE23-0010) and partly by ANR FiT LabCom.

\bibliographystyle{abbrv}
\bibliography{bibliography}

\appendix
\section{NP-hardness of the \DPM{} problem}\label{appendix:dpm}

Given a graph G and an order $\prec$ on the vertices of
G, we define $succ_{\prec}(u)$ (respectively $pred_{\prec}(u)$) as the set of neighbors $v$ of $u$ such that $u\prec v$ (resp. $v\prec u$). For any subset of vertices $W$, we note $\DPM_\prec(W)=\sum_{u\in W} |succ_{\prec}(u)|\cdot |pred_{\prec}(u)| $.
 Using this definition we formalize the following problem:

\begin{problem}[\DPM{}]
Given an undirected graph $G=(V,E)$ and an integer $K$, is there an order $\prec$ on the vertices such that 
    $\DPM_\prec(V) \leq K$?
\end{problem}



\begin{problem}[\SAT{}]
Not-All-Equal Positive Three-Satisfiability.
Given a formula $\phi=c_1\land\dots\land c_m$ in  conjunctive normal form where each clause consists in three positive literals, is there an assignment to the variables satisfying $\phi$
such that in no clause all three literals have the same truth value?
\end{problem}

The \SAT{} problem is known to be NP-complete by Schaefer's dichotomy theorem \cite{schaefer1978complexity}.
We will show that this problem can be reduced to the \DPM{} problem, thus proving that \DPM{} is NP-hard.
Note that a proof was given in-hard~\cite{stack2017dpm} but, as far as we know, it has never been published. We give a new simpler proof of the following theorem:

\begin{theorem}
\DPM{} is NP-hard.
\end{theorem}

\begin{definition}
Let $\phi$ be an instance of \SAT{} with variables
$x_1,\dots,x_n$ and clauses $c_1,\dots,c_m$, where clause $c_j$ is of the form $l_j^1\lor l_j^2\lor l_j^3$.
We define a graph $G_\phi$ by creating three connected vertices $L_j^1,L_j^2,L_j^3$ representing the literals of each clause $c_j$; additionally, a vertex $X_i$ is created for each variable $x_i$ and connected to all the $L_j^a$ such that $l_j^a=x_i$. 
%
More formally, $G_\phi=(V_\phi, E_\phi)$ with:
\begin{itemize}
    \item $V_\phi = 
        \{ X_i\ |\ i\in \llbracket 1, n\rrbracket \} \cup 
        \{ L_j^1, L_j^2, L_j^3\ |\ j\in \llbracket 1, m\rrbracket \}$
    \item $E_\phi = 
        \left\{ \{L_j^1, L_j^2\}, \{L_j^1, L_j^3\}, \{L_j^2, L_j^3\}\ |\ j\in \llbracket 1, m\rrbracket \right\} \cup
        \left\{ \{ X_i, L_j^a \}\ |\ 
        x_i = l_j^a \right\}$
\end{itemize}

\vspace{5mm}
\begin{center}
\begin{tikzpicture}[->,>=stealth,shorten >=1pt,auto,node distance=2.8cm,semithick]
  \node (l1) at (2,3.7) {$L_j^1$} ;
  \node (l2) at (2,2.3) {$L_j^2$} ;
  \node (l3) at (4,3) {$L_j^3$} ;
  \node (x1) at (.5,3.7) {$\dots\ X_{i_1}$} ;
  \node (x2) at (.5,2.3) {$\dots\ X_{i_2}$} ;
  \node (x3) at (5.5,3) {$X_{i_3}\ \dots$} ;
  \draw[-] (l1) -- (l2) ;
  \draw[-] (l3) -- (l2) ;
  \draw[-] (l1) -- (l3) ;
  \draw[-] (l1) -- (x1) ;
  \draw[-] (l2) -- (x2) ;
  \draw[-] (l3) -- (x3) ;
\end{tikzpicture}
\end{center}
\end{definition}
\vspace{5mm}

\begin{proposition}[$\Longrightarrow$]
Given an instance $\phi$ of \SAT{} with $m$ clauses and the associated graph $G_\phi$, if $\phi$ is satisfiable then there exists an order $\prec$ on $V_\phi$ such that $\DPM_\prec(V_\phi)\leq 2m$.
\end{proposition}

\begin{proof}
Let $\phi$ be a satisfiable instance of \SAT{} with the above notations. Take a valid assignment and let us note $k$ the 
number of variables set to true. There exist indices $i_1,\dots,i_n$ such that $x_{i_1}, \dots, x_{i_k}=true$ and $x_{i_{k+1}},\dots,x_{i_n}=false$, and for each clause $c_j$, there are indices $t_j,a_j,f_j\in\{1,2,3\}$ such that $l_j^{t_j}=true$, $l_j^{f_j}=false$ and $l_j^{a_j}$ has any value.
%
Now construct the following order on $V_\phi$, so that true variables come first, then in each clause the false literal comes before the true one, and the false variables are at the end: 
%
\begin{align*}
    & X_1\prec\dots\prec X_k\qquad&
    & \text{True variables}\\
\prec\ 
    & L_1^{f_1}\prec \dots\prec L_m^{f_m}&
    & \text{False literals} \\
\prec\
    & L_1^{a_1}\prec \dots\prec L_m^{a_m}&
    & \text{Other literals} \\
\prec\
    & L_1^{t_1}\prec \dots\prec L_m^{t_m}&
    & \text{True literals} \\
\prec\
    & X_{k+1}\prec\dots\prec X_n&
    & \text{False variables}
\end{align*}

If a given variable $x_i$ is true, the associated vertex $X_i$ has only successors, if it is false it has only predecessors, so in both cases $\DPM_\prec(\{X_i\})=0$.
For a given clause $c_j$, the variable $l_j^{f_j}$ is false so the corresponding $X_i$ is a successor of $L_j^{f_j}$, which also has successors $L_j^{a_j}$ and $L_j^{t_j}$, but no predecessor. Similarly, $L_j^{t_j}$ has no successor; thus $\DPM_\prec(\{L_j^{f_j}, L_j^{t_j}\})=0$.
Now $L_j^{a_j}$ has one predecessor $L_j^{f_j}$, one successor $L_j^{t_j}$, and one neighbor $X_i$ that is a predecessor if $x_i$ is true, otherwise a successor; in both cases, $\DPM_\prec(\{L_j^{a_j}\})=2$. 
The only vertices with a non-negative cost are the $L_j^{a_j}$, so the sum over all $m$ clauses gives $\DPM_\prec(V_\phi)=2m$.
\end{proof}

\begin{proposition}[$\Longleftarrow$]
Given an instance $\phi$ of \SAT{} with $m$ clauses and the associated graph $G_\phi$, if there exists an order $\prec$ on $V_\phi$ such that $\DPM_\prec(V_\phi)\leq 2m$ then $\phi$ is satisfiable.
\end{proposition}

\begin{proof}
Conversely, consider an order $\prec$ on $V_\phi$ such that $\DPM_\prec(V_\phi)\leq 2m$.
%
For all $j$, define $f_j,a_j,t_j\in\{1,2,3\}$ such that $L_j^{f_j}\prec L_j^{a_j}\prec L_j^{t_j}$; then $L_j^{a_j}$ has one successor, one predecessor, and one other neighbor $X_i$, so its cost is 2. As $G_\phi$ contains $m$ such independent triangles, $\DPM_\prec(\{L_1^{a_1},\dots,L_m^{a_m}\})=2m$.
To ensure $\DPM_\prec(V_\phi)\leq 2m$, all the other vertices must have either only predecessors or only successors. 
If vertex $X_i$ has successors only, assign $x_i$ to true; if $X_i$ has predecessors only, assign $x_i$ to false.
For all $j$, $L_j^{f_j}$ has at least 2 successors ($L_j^{a_j}$ and $L_j^{t_j}$) so its corresponding $X_i$ has to be a successor, which means $x_i=l_j^{f_j}$ is false; similarly, $l_j^{t_j}$ is true. Each clause thus has one true and one false literal, so $\phi$ is satisfied.
\end{proof}

\section{NP-hardness of the \DPP{} problem}\label{appendix:dpp}

\newcommand{\curProb}{\texorpdfstring{$C^{++}$}{C++}}
\newcommand{\costProb}{weighted-\curProb}
\newcommand{\G}{G}
\newcommand{\V}{V}
\newcommand{\E}{E}

\paragraph{Order of elimination.}

In the main text of the paper, we search for a permutation 
$\pi$ but the only important aspect of the permutation is that
it defines  an order on the vertices. For this NP-hardness proof, it will help
think of the following equivalent but more ``intuitive" formulation 
of the problem:  we are looking for an order $\prec$ minimizing the
cost function of interest. 
We can think of the order as an order in which we eliminate vertices 
and each time we eliminate a vertex with an outdegree $d$ we pay a cost of
$d^2$ and the cost of an order is the cost of  eliminating all vertices.

For the formulation using orders it will help to look at the set of neighbors of $u$
appearing after $u$ in the order $\prec$, which we denote $succ_\prec(u)$ for an order~$\prec$.
Therefore $|succ_\prec(u)|^2$ is
the cost that we pay when we remove $u$, which allows us to reformulate
the \curProb{} problem as:

\begin{problem}[\curProb{}]
For a given undirected graph
    $\G{}=(\V,\E)$ and an integer $K$, does
    there exist an order $\prec$ of the vertices such that 
    $\sum_{u\in \V} |succ_{\prec}(u)|^2 \leq K ?$
\end{problem}


\paragraph{The \costProb{} problem.}

For the sake of simplicity our proof of completeness will rely on a second novel
problem, the \costProb{} problem, and we will show that \curProb{} is NP-complete
by exhibiting first a reduction between \curProb{} and \costProb{} and then a second
reduction between \costProb{} and the Set Cover problem (a well-known
NP-complete problem). We now present the \costProb{} problem:

\begin{problem}[\costProb{}]
Given an undirected graph
    $\G{}=(\V,\E)$, a vertex-weighting function
    $w:\V\rightarrow \mathbb{N}$ and an integer $K$, does
    there exist an order $\prec$ of the vertices such that {$\sum_{u\in
      \V} (|succ_{\prec}(u)|+w(u))^2 \leq K$?}
\end{problem}

\paragraph{Terminology.}
Given a graph $G$ with the vertex weighting function $w$ and an order $\prec$, the \emph{cost} is the
function {$\sum_{u\in
      \V} (|succ_{\prec}(u)|+w(u))^2$} applied to
the graph with that order. The {\emph{optimal cost}} of a graph is
the minimal {cost} achievable by any order.
Notice that an instance of the weightless problem can be viewed as an
instance of the weighted problem where all weights are 0.

\subsection{Optimality criteria for orders.}

One difficulty of the reduction proofs is to show
that an order necessarily behaves in a controlled way. We see in this 
section several criteria that ensures that some order has an optimal cost.

We define the notion of multiset of costs that will help expressing optimality criteria for orders.
Given a graph $\G{}$ and an order $\prec$, the multiset
of costs $MC(\G{},\prec)$ is the multiset composed of the
\ltok{$\left( |succ_\prec(u)| + w(u) \right) $} for each vertex $u$ in $\G{}$.
The \emph{linear cost} of a multiset $M$ is the sum of elements in the multiset,
i.e., $\sum_{c\in M} c$.
The \emph{cost} (or \emph{squared cost}) of a multiset $M$ is $\sum_{c\in M} c^2$.

\begin{property}
  For a graph $\G{}$ (weighted or not) the size and the linear
  sum of  the multiset $MC(\G{},\prec)$ does not depend on the
  order $\prec$.
\end{property}

\begin{proof}
By definition, the size of $MC(\G{},\prec)$ is $|V|$, the number of vertices in 
$\G{}$, and its linear cost is 
%
%
\ltok{$\sum_{c\in M} c = \sum_{u\in V} |succ_\prec(u)| + w(u) = |E| + \sum_{u\in V} w(u)$}.
%
\end{proof}

Note that this allows us to talk about the \emph{linear cost} of a 
graph $\G{}$ as the linear cost of any multiset of costs, corresponding to any of its order.

\begin{property}
  \label{prop:optOrderW}
When there exists some $d\in\mathbb{N}$ such that $MC(\G{},\prec)$
contains only the values $d$ and $d+1$ then the order $\prec$ is
optimal.
\end{property}

\begin{proof}
Let us consider an order $\prec$ as described above:
it only contains $a$ times $d$ and $b$ times $d+1$ for some $d$,
and let us consider any optimal order $\prec'$ of $\G{}$.
Suppose that the multiset $MC(\G{},\prec')$ contains $e$ and $f$
such that $e<f-1$. Then replacing them by $e+1$ and $f-1$ reduces 
the cost because $(e^2+f^2)-\left((e+1)^2+(f-1)^2\right)=2(f-1-e)>0$. 
By iteratively applying this operation we end up with a multiset $M$ 
that has the same size and the same linear cost 
but a lower cost than $MC(\G{},\prec')$ and contains $a'$ times $d'$ and $b'$ times
$d'+1$ for some $d'$.


Without loss of generality we can suppose that there is at least one $d$
in $MC(\G{},\prec)$ which means that $d$ is \ltok{the quotient of the Euclidean division of the linear cost by $ |V| $}.
For the same reason, $d'$ is \ltok{the quotient of the Euclidean division of the linear cost of $MC(\G{},\prec')$  by $ |V| $}.
%
Because the linear cost of $MC(\G{},\prec)$ does not depend 
on $\prec$ this proves that $d'=d$
\ltok{as well as $ a=a'$ and $ b=b'$, which in turn implies that the costs of $MC(\G{},\prec)$ and $MC(\G{},\prec')$ are similar, and thus $\prec$ is optimal.}
%
\end{proof}

While the property above is true for any graph (weighted or not) it is
not really useful for weightless graphs because, in a weightless graph,
the last vertex $u$ that we eliminate in the order $\prec$ has 
$|succ_\prec(u)|=0$, and more generally the vertex $u_i$ which is ordered
in the $i$-th position from the end,
has $|succ_\prec(u_i)| < i$. The following property handles this case:

\begin{property}
  \label{prop:optOrder}
  For a weightless graph, when there exists $d\in\mathbb{N}$ such that
  \mbox{$MC(\G{},\prec)$} contains all integers from $0$ to $d+1$ and
  at most once the integers $0$ to $d-1$, then $\prec$ is an optimal
  order.
\end{property}

\begin{proof}
The proof of optimality is similar to the proof of
property~\ref{prop:optOrderW}: when the property does not hold, 
we can find two elements $v_i$ and $v_j$ with $v_i+2 \leq v_j$ and
we can diminish the cost by setting $v_i=v_i+1$ and $v_j=v_j-1$.
\end{proof}

Finally, let us introduce the notion of \emph{marginal cost} for a multiset.

\begin{definition}[Marginal cost]
We introduce the \emph{marginal cost} to measure how much
a  multiset deviates from the optimal repartition (as given in 
property~\ref{prop:optOrderW}). Formally,  given a multiset $M$ of size $n$, we can
compute $d$ such that the linear cost of $M$ is  $n\times d + v$ 
where $1 \leq v\leq n$. 
The \emph{marginal cost} $c_m$ of $M$ is then:
$$c_m = \sum_{u\in M} \max\Big(0,u-(d+1)\Big)$$
\end{definition}


Note that we can equivalently define the marginal cost of $M$ as
$\sum_{u\in M} u  - \sum_{u\in M} min(d+1,u)$.
We know from property~\ref{prop:optOrderW} that the multiset $M'$ that 
minimizes the cost with the same linear cost and the same size only contains
$d$ and $d+1$ (with at least one $d+1$ since $v>0$). In other words the
marginal cost counts the number of elements larger than needed and 
how much they go over the average cost: if we have a $d+2$ it counts for
$1$, if we have a $d+5$ it counts for $4$, etc.

Note that the marginal cost cannot be used  directly to decide if an order is 
optimal.  Indeed, consider the two following multisets: $M$ composed of 
nine times the value 10 and one time the value 11 and $M'$ composed of 
nine times the values 11 and one time the value 2. They
have the same size, the same linear cost and the same marginal cost 
(which is $0$) but $M$ has a lower squared cost than $M'$.

The following property describes the minimal cost among
all the multisets with the same size, the same linear cost and the same 
marginal cost:
\begin{property}
\label{prop:marginal}
Among all the multisets that have a size $n$, a 
linear cost of $d\times n+v$ with $1 \leq v \leq n$
and a marginal cost of at least $k$ (with $2k<v$),
then the ones achieving the minimal cost are
composed of $k$ times the value $d+2$, $v-2k$ times the
value $d+1$ and $(n-v+k)$ times the value $d$.
\end{property}

\begin{proof}
Take such a multiset $M$ with size $n$, linear cost of $d\times n+v$ and marginal cost of at least $k$, satisfying $2k<v$.
%
%
Suppose that $M$ contains a value $d-i$ with $i>0$. Because  the linear cost of $M$ is strictly larger than $d\times n$ we can find at least one 
value  $d+j$ with  $j>0$ such that diminishing $d+j$ to $d+j-1$ keeps the 
marginal cost above $k$.  

\ltok{
Indeed, in a first case, there is at least one value $d+1$ in $M$, and diminishing this value to $d$ does not affect the marginal cost $c_m$.
In the other case, let us consider the sum $\sum_{u\in M} min(d+1,u)$, by definition of the marginal cost, there are no more than $c_m$ elements in $M$ larger than $d+1$.
All other elements are at most $d$ with at least one at $d-i<d$.
So, we have $\sum_{u\in M} min(d+1,u) < nd + c_m$ and using $c_m = \sum_{u\in M} u  - \sum_{u\in M} min(d+1,u)$, we find that
$$ c_m > nd + v - (nd + c_m) \Rightarrow 2c_m > v$$
and as we have made the assumption that $v > 2k$, we have $c_m > k$. Consequently, we can also find in this case a value $d+j$ with  $j>0$ such that diminishing $d+j$ to $d+j-1$ keeps the marginal cost above $k$.
}



\ltok{
We can deduce from this observation that the multiset of size $n$, with a linear cost $d\times n+v$ with $1 \leq v \leq n$
and a marginal cost of at least $k$ (with $2k<v$) that achieves the minimal cost has $k$ times the value $d+2$, $v-2k$ times the value $d+1$ and $(n-v+k)$ times the value $d$.}
\end{proof}

This property can then be used to compare multisets and is summarized by
the following property:

\begin{property}
\label{prop:marginalOpt}
When $M$ has a marginal cost of $k$ then the cost of $M$ is at least $2k$ larger
than the balanced distribution (as given by property~\ref{prop:optOrderW}). 
This $2k$ bound is reached for the optimality criterion described in property~\ref{prop:marginal}.
\end{property}
\begin{proof}
As seen before the optimal can be reached by taking two values $i,j\in M$ with 
$i+2\leq j$ and changing them to $i+1$ and $j-1$. This balancing operation 
can reduce by at most $1$ the marginal cost but reduces the cost by $i^2+j^2-(i+1)^2-(j-1)^2=2(j-i)-2$ and since $j-i\geq 2$ this means the reduction
is at least $2$ and exactly $2$ when $i+2=j$. Since we need at least $k$ 
balancing operations to reach the optimal, this gives us at least a $2k$ reduction of the cost to reach the optimal.
Notice that when dealing with an  optimal multiset in the sense of property~\ref{prop:marginal} we only combine a $d$ with a $d+2$ which 
gives us the exact bound. Conversely if we are not in the case of 
\ref{prop:marginal} we will have to combine something below (or equal to) $d$ 
with something larger than $d+3$ or do a combination that does not diminish
the marginal cost (such as combining $d+1$ and $d+3$).
\end{proof}

\subsection{Reduction between \costProb{} and \curProb{}.}

Any instance of \curProb{} can be seen as an instance of \costProb{} where
the weights are set to 0. For that purpose, the idea is to take a vertex $u$ with some non null 
weight $w(u)$, link $u$ to $w(u)$ vertices $v_1, \dots, v_{w(u)}$ and make sure 
that we can guarantee that  $u$ appears before all the 
$v_1, \dots , v_{w(u)}$ in any optimal order.
We will thus exhibit a family of graphs to create such $v_i$ vertices before 
showing that these $v_i$ vertices can always appear after $u$ in the order. Finally we will  
prove the full reduction.

\subsubsection{The \texorpdfstring{$L_d$}{Ld} family of graphs.}

Let us consider the graph $L_d$ parameterized by $d\in\mathbb{N}$ that
contains a $(d+1)$-clique $ K^d$ composed of the vertices $K^d_0, \dots, K^d_d$,
one vertex $e_d$ that has $d$ neighbors $v^d_1 \dots v^d_d$ and such
that there is an edge between each $v^d_i$ and each vertex of $K^d$.
Consequently, there are three types of vertices in $L_d$: the vertex 
$e_d$, the vertices of type $V$ (the $(v^d_i)_i$) and the vertices of 
type $K$ (the $(K^d_i)_i$).
Here is a depiction of $L_d$: 

\vspace{5mm}
\begin{center}
\begin{tikzpicture}[->,>=stealth,shorten >=1pt,auto,node distance=2.8cm,semithick,scale=.8]

  \node (n) at (0,0) {$e_d$} ;

  \node (v1) at (2,1.5) {$v^d_1$} ;
  \node (v2) at (2,0.75) {$\vdots$} ;
  \node (v3) at (2,0) {$v^d_i$} ;
  \node (v4) at (2,-0.75) {$\vdots$} ;
  \node (v5) at (2,-1.5) {$v^d_d$} ;

  \draw[-] (n) -- (v1) ;
  \draw[-] (n) -- (v3) ;
  \draw[-] (n) -- (v5) ;

  \node (k2) at (9.5,2) {$K^d_0$} ;
  \node (k3) at (8,1) {$K^d_1$} ;
  \node (kj) at (9,0) {$\dots$} ;
  \node (k1) at (8,-1) {$K^d_j$} ;
  \node (k4) at (9.5,-2) {$K^d_d$} ;
  
  \draw[-] (k1) to [bend left=10] (v1) ;
  \draw[-] (k1) -- (v3) ;
  \draw[-] (k1) to [bend right=10] (v5) ;

  \draw[-] (k2) to [bend right=10] (v1) ;
  \draw[-] (k2) to [bend right=15] (v3) ;
  \draw[-] (k2) to [bend right=20] (v5) ;

  \draw[-] (k3) to [bend left=10] (v1) ;
  \draw[-] (k3) -- (v3) ;
  \draw[-] (k3) to [bend right=10] (v5) ;

  \draw[-] (k4) to [bend left=20] (v1) ;
  \draw[-] (k4) to [bend left=15] (v3) ;
  \draw[-] (k4) to [bend left=10] (v5) ;
  
  \draw[-,bend right] (k4) -- (k3) -- (k2) -- (k1) -- (k4) -- (k2) -- (k3) -- (k1);
  
\end{tikzpicture}
\end{center}
\vspace{5mm}

\paragraph{Best cost $C_d$ for $L_d$.}

In the weightless case,
the best cost $C_d$ for $L_d$ is induced by the order that starts with
$e_d$ followed by the $v_i^d$ nodes and finally by the $K^d_i$ nodes.
Indeed, in that case the cost is $d^2$ for $e_d$, $(d+1)^2$ for each
$v^d_i$ and $i^2$ for $K^d_i$ (supposing we start with $K^d_d$ and
end with $K^d_{0}$).
This is optimal by virtue of property~\ref{prop:optOrder}.

\paragraph{Best cost for $L_d$ with a weight $1$ on $e_d$.}

If we add a weight 1 on $e_d$ then the best cost can be achieved
with the same order but this time the cost of $e_d$ is increased from
$d^2$ to $(d+1)^2$ which means an increase of $2d+1$.
In other words, in that case, the best cost is $C_d +2d+1 $.
Note that here, the optimality cannot be deduced directly from
property~\ref{prop:optOrder} as the property only applies to
weightless graphs. However we prove that there is an optimal order
starting with $e_d$.

For that, consider any order $\prec$ and let us show that $\prec$ can
always can be improved to an order that places the vertex $e_d$ in first position.


The order $\prec$ ranks three types of vertices: $e_d$, $V$ nodes and $K$ nodes, according to the description above.
Let us first suppose that there is a vertex of type $K$ before a vertex of type
$V$ before the vertex $e_d$. In that case the first $i$ vertices are
of type $V$ (we can have $i=0$), then we have $j+1$ vertices of type $K$
and then one vertex of type $V$. Let us consider how the cost changes by
exchanging this last $K$ with this last $V$, i.e., to change from
$V^iK^{j}KV$ to $V^iK^j V K$. It is clear that the cost changes only for the
exchanged $V$ and $K$. Before the exchange the cost of $V$ was
$(d-j)^2$ and after it is $(d-j+1)^2$ whereas for $K$ it was
$(2d-i-j)^2$ and after it is $(2d-i-j-1)^2$.
Overall, if $\Delta C$ is the difference between the cost before and the cost after the exchange, we have:
\begin{align*}
\Delta C & = (2d-i-j)^2 -(2d-i-j-1)^2 \\
& \ \ \ \ + (d-j)^2-(d-j+1)^2\\
& = 2d-2i-4\\
& = 2 (d-i-2)
\end{align*}

Therefore, unless $i+1=d$, the cost decreases which means
that we can always move the $V$ vertices at the beginning except for
maybe one to improve the cost of the order. 
In the end we have that the beginning of an optimal
sequence can be restricted to the form $V^iK^l e_d$ or $V^{d-1}K^lV e_d$. 
In the first case, transforming $V^iK^l e_d$ into $e_d V^iK^l$ decreases the score by $(i^2+i)$.
In the second case, transforming $V^{d-1}K^lV e_d$ into $e_d V^{d-1}K^lV$ decreases the score by $d^2+d-2l$ (which is $\geq 0$ 
because $l\leq d$).
Thus, we can move in all cases $e_d$ at the beginning of the order to improve the related cost.

We have proved that the best cost can be achieved by placing $e_d$ at the beginning of the order.
As the best cost $C_d$ for the rest of the order is unaffected by the cost of the elimination of $e_d$ first,
we have that the best cost is $C_d+2d+1$.


\subsubsection{Partitioned graphs.}

Let us consider a graph $\G{}$ composed of two subgraphs
$\G{}_1$ and $\G{}_2$ plus exactly one edge
$\{e_1,e_2\}$ with $e_1\in V_1$ and $e_2\in V_2$. Any order $\prec$ on
$V$ induces an order on $V_1$, an order on $V_2$ and an order between
$e_1$ and $e_2$. 
If another order $\prec'$ induces the same order on
$V_1$, the same order on $V_2$ and the same order between $e_1$ and
$e_2$, it has the same cost as $\prec$. 
Therefore, an optimal
order for $G$ can be seen as either an optimal for $\G_1$ and an
optimal order for $\G_2$ where we add a weight of 1 on $e_2$ (if $e_2$ precedes $e_1$), or an
optimal order for $\G_2$ and an optimal order for $\G_1$ where we
add a weight of 1 on $e_1$ (if $e_1$ precedes $e_2$).

As a result, we obtain the following property:

\begin{property}
  \label{prop:split}
  If adding a weight 1 on $e_1$ in $\G_1$ increases the best cost of 
  $\G_1$ by $x$ and if adding a weight 1 on $e_2$ increases the 
  best cost of $\G_2$ by at most $x$, then the best cost of $\G{}$ is
  equal to the best cost of $\G_1$ plus the best of $\G_2$ where we add
  a weight of $1$ on $e_2$.
\end{property}

\subsubsection{Finishing the reduction.}
\begin{property}
  Let $(\G{},K)$ be an instance of the weighted problem, we can
  compute an equivalent instance of the weightless problem in
  a time polynomial in the number of edges and vertices in $\G$ 
  plus the sum of weights in $\G$.
\end{property}

\begin{proof}
If all the weights in $\G{}$ are zeros, the result is immediate.
Let us suppose that there is a vertex $u$ with a
weight $w(u)>0$ and a degree $d-w(u)$. Let us consider 
the graph $\G{}'$ composed of  $\G{}$ but where the weight 
of $u$ is reduced by 1 plus a fresh  copy of $L_d$  and an
edge between $u$ and $e_d$. 
We claim that the best cost of $\G{}'$ is lower than $K+C_d$ if 
and only if the best cost of $\G{}$ is lower than $K$.

Indeed, we have shown that the graph $L_d$ is such that adding a weight 1 on
$e_d$ increases the best cost from $C_d$ to $C_d+2d+1$. We also know
that the sum of the degree of $u$ plus its cost is $d$ therefore for
any order $\prec$ adding a weight 1 on $u$ increases the cost of $\prec$ 
by at most $2d+1$. 
\ltok{By applying property~\ref{prop:split} where $\G{}$ has the role of $\G_2$ ($u$ is $e_2$) and $L_d$ of $\G_1$ ($e_d$ is $e_1$) and $ x=2d+1$,
we obtain that the best cost of $\G{}'$ is the best cost of $\G{}$ where node $u$ has a weight increased by 1.
In other words, it is equivalent in terms of best cost to handle the graph $\G{}$ or to handle the graph $\G{}'$ where the weight of node $u$ has been decreased by 1 unit. 
}
%

By applying $\sum_u w(u)$ times this property we obtain \ltok{an instance $(\G{}',K')$ of the weightless problem}
which is equivalent to the \ltok{instance of the weighted problem} $(\G{},K)$. This new instance
has $\sum_u w(u)\times |L_{deg(u)+w(u)}|$ more vertices than the original one,
but this is still polynomial in the size of $\G{}$ plus the sum of 
weights and the  resulting instance can be computed in polynomial time.
\end{proof}

This proves that if the \costProb{} problem is strongly NP-hard, the
\curProb{} problem is also NP-hard.

\subsection{Reduction between the \costProb{} and Set Cover.}

Our reduction for the weighted case will be a \emph{strong reduction},
meaning the version of the problem where the weights are polynomial in 
the size of the graph is still NP-hard.
It will be based on the \emph{Set cover problem}. We recall
here the definition of this problem and invite the reader to check the
literature for a proof of its NP-completeness:


\begin{problem}[Set cover]
  Given two integers $n,k$,
  \ltok{we denote $U$ the set of elements $\{1,\ldots,n\}$.
  Let $P$ be a set of sets of elements of $U$, does there exist a subset $P'\subset P$ of size $k$ such that 
  $\displaystyle{\cup_{S\in P'} S = U}$}?
\end{problem}


Let us fix an instance $(P,n,k)$ of the Set Cover problem asking whether we can find $k$ sets $S_1, \dots, S_k$ in $P$ such that $S_1
\cup \dots \cup S_k = \ltok{U}$. We suppose, without loss of
generality, that the instance is not trivial in the sense that $|P|
\geq k$ (there are at least $k$ sets in $P$), $\cup_{S\in P} S = \ltok{U}$ 
(each integer in $\ltok{U}$ is contained in at least one $S\in P$) 
and that all sets $S\in P$ are such that
$S\subseteq \ltok{U}$.

Let us exhibit a weighted graph $\G$ and a value $V$ such that best 
cost for $\G$ is less than $V$ if and only if $\ltok{\{1,\ldots,n\}}$ can be covered 
with $k$ sets from $P$.


\subsubsection{\ltok{Construction of a \costProb{} instance from a Set Cover instance.}}

Our reduction will provide a graph $\G{}$ with a weight function $w$ 
depending on a parameter $d$ such that the Set Cover instance has a solution
if and only if the best order has a multiset of costs containing 
at most $k$ values $d+2$ and all other values are either $d$ or $d+1$ (we will 
explicit later the values of $V$ and $d$). 

\paragraph{Vertices of $\G{}$.}
In $\G$ the vertices are: a special vertex $A$, $n$ vertices $e_1, \dots, e_n$ one for
each $i\in \ltok{\{1,\ldots,n\}}$, $\ell$ vertices, $s_1, \dots, s_\ell $ with one vertex $s_j$ for each set
$S_j\in P$ and finally three vertices $a^i_j,b^i_j,c^i_j$ for each $i\in S_j$. 

\paragraph{Edges of $\G{}$.} 

The vertex $A$ has an edge with all vertices of the form $s_j$ or $e_i$.
For a pair $(i,j)$ with $i\in S_j$, both $a^i_j$ and $b^i_j$ have an edge with $s_j$ and $c^i_j$; in turn $c^i_j$ has an edge with $e_i$.

Overall the graph $G$ looks like this:  





  
  
  
  


\begin{center}
\begin{tikzpicture}[->,>=stealth,shorten >=1pt,auto,node distance=2.8cm,semithick,scale=.9]

  \node (A) at (-1,2.5) {$A$} ;

  \node (s1) at (1,4) {$s_1$} ;
  \node (s2) at (2,4) {$\dots$} ;
  \node (s3) at (3,4) {$s_j$} ;
  \node (s4) at (4,4) {$\dots$} ;
  \node (s5) at (5,4) {$s_\ell$} ;

  \draw[-] (A) to [bend left] (s1) ;
  \draw[-] (A) to [bend left=60] (s3) ;
  \draw[-] (A) to [bend left=75] (s5) ;

  \node (a1) at (2,3) {$a^i_j$} ;
  \node (b1) at (4,3) {$b^i_j$} ;
  \node (c1) at (3,2) {$c^i_j$} ;
  
  \node (f1) at (1, 2.5) {$\vdots$};
  \node (fn) at (5, 2.5) {$\vdots$};
  
  \draw[-] (s3) to [bend left=20] (a1) ;
  \draw[-] (s3) to [bend right=20] (b1) ;
  \draw[-] (a1) to [bend left=20] (c1) ;
  \draw[-] (b1) to [bend right=20] (c1) ;
  
  \node (n1) at (1,1) {$e_1$} ;
  \node (n2) at (2,1) {$\dots$} ;
  \node (n3) at (3,1) {$e_i$} ;
  \node (n4) at (4,1) {$\dots$} ;
  \node (n5) at (5,1) {$e_n$} ;
  
  \draw[-] (c1) -- (n3) ;

  \draw[-] (n1) to [bend left] (A) ;
  \draw[-] (n3) to [bend left=60] (A) ;
  \draw[-] (n5) to [bend left=75] (A) ;

\end{tikzpicture}
\end{center}

\paragraph{Weights of vertices in $\G{}$.}

Recall that the cost of a vertex is the sum of its degree and its weight.
In $G$, we set the weights so that each vertex has a cost of $d+2$, except for the $c^i_j$ which have a cost of $d+3$ and the vertex $A$ which has a cost of $d+1+n+k$.  
Parameter $d$ needs to be large enough so that all weights are positive
%
%
\ltok{This is not constraining for vertices $a^i_j$, $b^i_j$ and $c^i_j$.
Vertices $s_j$ have degree 1 plus twice the number of $i$ appearing in set $S_j$, i.e., $1+2 \times |S_j|$, so it suffices that $d>2\times |S_j|$ for all $S_j\in P$.
Vertices $e_i$ have degree 1 plus the number of $S_j$ sets where $i$ can be found, which is at most $1+\ell$.
Vertex $A$ has degree $ \ell + n$.
Having the additional condition $ d>\ell$ is sufficient to guarantee the constraint on $A$ and $e_i$ vertices.}

\paragraph{Value of $V$.}

As we will show, when there is a Set Cover with $k$ sets then 
we have an order $\prec$ for $\G$  such that $MC(\G,\prec)$ contains $k$ 
times the value $d+2$ (corresponding to the $k$ selected sets),  $\sum_{S\in P} |S| -  n$ times the value $d$ and 
all the other values are $d+1$.
It implies that the cost  $V=k (d+2)^2 + (\sum_{S\in P} |S| -  n) d^2 + r (d+1)^2$,
where $r$ is the number of vertices in $\G$ minus $k$ and minus
$(\sum_{S\in P} |S| -  n)$.

Note that, per property~\ref{prop:marginal}, this value $V$ corresponds to
the minimal cost for an order that has a marginal cost of $k$.
Conversely, we will show that if there is a solution with a marginal
cost of $k$ or less then there is a Set Cover with $k$ sets, proving that 
it is a reduction.
Note that this converse direction is stronger than what is needed as there
exists multisets with a marginal cost of $k$ that do not match the minimal 
cost.

\ltok{
The general intuition underlying the equivalence between a solution (if any) of the Set Cover problem and a solution of the corresponding \costProb{} problem  is the following.
The first $k$ vertices $s_j$ selected in the elimination order correspond to the $S_j$ sets that cover $U$.
Indeed, each of these vertices generate exactly a marginal cost of 1 and all other nodes according to the elimination order will not generate any marginal cost if we can eliminate all $e_i$ nodes without adding any marginal cost.
This condition is met if deleting the $k$ first $s_j$ nodes allows to decrease the cost of all $e_i$ nodes by (at least) 1 unit, which means that we have deleted at least one triplet $a^i_j$, $b^i_j$, $c^i_j$ related to node $e_i$.  
If so, we have found an elimination order with cost $V$ as well as $k$ sets $S_1 \ldots S_k \in P$ which cover $U$. 
}

\subsubsection{Proof that a solution to Set Cover implies a solution to \curProb{}.}

Suppose that we have a solution to Set Cover with the sets 
$S_{j_1}, \dots, S_{j_k}$. Let us prove that our graph $\G{}$
has an elimination order where the cost of each vertex  is $d$ or $d+1$
or $d+2$ but with only $k$ vertices with cost $d+2$.

The elimination order can be built by having $j$ going through $j_1,\dots,j_k$.
For each $j$ value, we eliminate first $s_{j}$ for a cost of $d+2$,
then we go through $i \in S_j$ and eliminate the corresponding $a_j^i$ and $b_j^i$ vertices
(both at cost $d+1$ once $s_j$ has been removed). Then we eliminate
$c_j^i$ (for a cost of $d$ if $e_i$ is already eliminated and $d+1$  otherwise).
Finally, if $e_i$ has not yet been eliminated by a previous $j$ value, we eliminate it for a cost of $d+1$. 

Once we have done this, the vertex $A$ has lost $k+n$ neighbors: all the $e_i$
and the $k$ vertices $s_j$ that we have selected. Its remaining cost is $d+1$ so we
eliminate it, which in turn means that all the remaining $s_j$ have a cost of $d+1$
and we can eliminate them all (with their $a^i_j$, $b^i_j$ and $c^i_j$ attached).

Overall the cost of this elimination order is exactly $V$.

\subsubsection{Proof that a solution to \costProb{} implies a solution to Set Cover.}

Suppose that we have an order $\prec$ such that the total cost is below $V$. Since $V$ is the optimal cost for a marginal cost of $k$, the order $\prec$ cannot have a marginal cost higher than $k$ otherwise its cost would be higher than $V$ (see 
property~\ref{prop:marginalOpt}).
Knowing that $\prec$ has a marginal cost of at most $k$, we will extract a solution to 
the corresponding Set Cover instance.

First we notice that when $A$ is eliminated, its cost is 
$d+1+k-E_s+E_n+R_n$ where $E_s$ the number of $s_i$ eliminated,
and $R_n$ is the number of $e_i$ remaining. However, as long as $A$ 
is not  eliminated, $E_s$ is less than (or equal to) the marginal cost of all
the vertices eliminated before $A$. Indeed, if $s_j$ is eliminated while $A$
is still present it is because we have paid a marginal cost at least $1$ to
eliminate directly one of $s_j$ or $a^i_j$ or $b^i_j$ or $c^i_j$ for
$i\in S_j$.
That is true because if  $A$ is present, then all those vertices have a cost
of $d+2$ except $c^i_j$ that has a cost $d+3$ or $d+2$ depending on
whether $e_i$ is eliminated or not yet.

Overall when we eliminate $A$, we pay a marginal cost of $(k-E_s)+R_n$ 
where $E_s$  is the number of sets eliminated and $R_n$ is the number of 
integers not  yet eliminated. The marginal cost of the order $\prec$ is at
least the marginal cost of all vertices eliminated before $A$ plus the
marginal cost for $A$. Because the marginal cost of vertices removed before
$A$ is at least $E_s$, by adding the marginal cost of $A$ we get a marginal
cost larger than $E_s+ (k-E_s)+R_n = k+R_n$  which can be equal to $k$ only
if $R_n=0$ which means that all vertices $e_i$ corresponding to integers \ltok{$\{1 \ldots n \}$} have been
eliminated. 
\ltok{Note that if a vertex $e_i$ is  directly eliminated without eliminating first a vertex $s_j$ and a triplet $a^i_j$, $b^i_j$, $c^i_j$ then   we have to add a marginal cost of $1$ specifically for this vertex $e_i$.}
%
%
But in that case, it means that the marginal cost of all vertices before $A$ 
includes the cost of removing this $e_i $ which means that we cannot have 
an overall marginal cost of $k$. Combining everything we get that if we have
an order that has a marginal 
cost of $k$ and \ltok{thus a cost of at most $V$}, then we have $k$ sets  $S_{i_1}, \dots, S_{i_k}$ covering all 
integers in $\ltok{U}$.

\end{document}